\documentclass[english]{lipics-v2021}
\nolinenumbers

\usepackage[utf8]{inputenc}

\usepackage{microtype,hyperref,graphicx}
   \hypersetup{%
      breaklinks,%
      colorlinks=true,%
      urlcolor=[rgb]{0.25,0.0,0.0},%
      linkcolor=[rgb]{0.5,0.0,0.0},%
      citecolor=[rgb]{0,0.2,0.445},%
      filecolor=[rgb]{0,0,0.4},
      anchorcolor=[rgb]={0.0,0.1,0.2}%
   }

\usepackage[linesnumbered,boxed]{algorithm2e}

\newtheorem{subproblem}[theorem]{Subproblem}

\title{Faster Algorithms for Reverse Shortest Path in Unit-Disk Graphs and Related Geometric Optimization Problems:\\ Improving the Shrink-and-Bifurcate Technique}

\titlerunning{Faster Algorithms for Reverse Shortest Path in Unit-Disk Graphs}

\author{Timothy M. Chan}{Siebel School of Computing and Data Science, University of Illinois at Urbana-Champaign, USA}{tmc@illinois.edu}{https://orcid.org/0000-0002-8093-0675}{Work supported by NSF Grant CCF-2224271.}
\author{Zhengcheng Huang}{Siebel School of Computing and Data Science, University of Illinois at Urbana-Champaign, USA}{zh3@illinois.edu}{https://orcid.org/0009-0006-0974-8818}{}

\authorrunning{T.\,M. Chan and Z. Huang}


\newcommand{\OO}{\tilde{O}}
\newcommand{\OOO}{O^*}
\newcommand{\dd}{\tilde{d}}

\newcommand{\ceil}[1]{\lceil#1\rceil}
\newcommand{\etal}{\emph{et al.}}

\newcommand{\kkappa}{\delta}

\def\Frechet{Fr\'{e}chet}
\def\Jejcic{Jej\v{c}i\v{c}}

\begin{document}

\maketitle

\begin{abstract}
In a series of papers, Avraham, Filtser, Kaplan, Katz, and Sharir (SoCG'14), Kaplan, Katz, Saban, and Sharir (ESA'23), and Katz, Saban, and Sharir (ESA'24) studied a class of geometric optimization problems---including reverse shortest path in unweighted and weighted unit-disk graphs, discrete Fr\'{e}chet distance with one-sided shortcuts, and reverse shortest path in visibility graphs on 1.5-dimensional terrains---for which standard parametric search does not work well due to a lack of efficient parallel algorithms for the corresponding decision problems. The best currently known algorithms for all the above problems run in $\OOO(n^{6/5})=\OOO(n^{1.2})$ time (ignoring subpolynomial factors), and they were obtained using a technique called \emph{shrink-and-bifurcate}. We improve the running time to $\OO(n^{8/7}) \approx O(n^{1.143})$ for these problems. Furthermore, specifically for reverse shortest path in unweighted unit-disk graphs, we improve the running time further to $\OO(n^{9/8})=\OO(n^{1.125})$.

%    In this paper, we study a range of geometric optimization problems, including reverse shortest path in unweighted and weighted disk graphs, discrete \Frechet{} distance with one-sided shortcuts, and reverse shortest path in visibility graphs on terrains. 
\end{abstract}

\section{Introduction}

\emph{Parametric search} is one of the most well-known and powerful techniques for solving geometric optimization problems~\cite{AgarwalS98}. First introduced by Megiddo~\cite{Megiddo83} in the 80s, the technique has since found countless applications in computational geometry and beyond. To search for the optimal value $r^*$ of an optimization problem, we first solve the decision problem: given an input value $r$, decide whether $r^*\le r$.  The idea is to simulate an algorithm $\mathcal{A}$ for the decision problem\footnote{
The algorithm $\mathcal{A}$ that we are simulating need not solve the decision problem, although it often is the case, so long as the optimal value $r^*$ is one of its ``critical values''.
} with the input value being the unknown $r^*$ itself; whenever $\mathcal{A}$ makes a comparison with the unknown $r^*$, we resolve the comparison by
calling the decision algorithm. Megiddo showed that if $\mathcal{A}$ is parallelizable with polylogarithmic number of rounds, then the simulation can be done efficiently with only a polylogarithmic factor increase in the total running time.

Over the years, a number of refinements and alternatives to parametric search have been proposed (e.g., Cole's trick~\cite{Cole87}, matrix searching~\cite{FredericksonJ82,FredericksonJ84}, expander-based approaches~\cite{KatzS97}, Chan's randomized optimization technique~\cite{Chan99} and implicit linear programming technique~\cite{Chan04,ChanHJ22}), but these alternatives are more specialized and are mainly about lowering the extra polylogarithmic factors in the running time.

In the last couple of years, Avraham, Filtser, Kaplan, Katz, and Sharir~\cite{AvrahamFKKS15}, Wang and Zhao~\cite{WangZ23}, Kaplan, Katz, Saban, and Sharir~\cite{KaplanKSS23},  Katz, Saban, and Sharir~\cite{KatzSS24}, and Agarwal, Kaplan, Katz, and Sharir~\cite{AgarwalKKS24} have identified and explored a natural class of geometric optimization problems (listed below) where parametric search does not do too well. For these problems, efficient near-linear-time decision algorithms exist, but they do not seem efficiently parallelizable (the decision version of these problems are related to single-source shortest paths, but breadth-first search (BFS) and Dijkstra's algorithm are inherently sequential). For many of the problems in this class, the search spaces are actually Euclidean distances defined by pairs of input points, and so parametric search could be replaced by binary search using an oracle for \emph{distance selection}, but unfortunately, known distance selection algorithms requires near $O(n^{4/3})$ time in the plane~\cite{AgarwalASS93,KatzS97,WangZ23Improved,ChanZ24}.

To overcome these issues, Avraham \etal~\cite{AvrahamFKKS15} introduced a new clever variant of the parametric search technique they termed \emph{shrink-and-bifurcate}, which was subsequently adopted by Kaplan \etal~\cite{KaplanKSS23}, Katz \etal~\cite{KatzSS24}, and Agarwal \etal~\cite{AgarwalKKS24}. The technique is able to solve most of these problems  with an unusual running time between near-linear and $n^{4/3}$: namely, $O^*(n^{6/5})$.\footnote{
In this paper, we use the $O^*$ and $\OO$ notation to hide $n^{o(1)}$ and $(\log n)^{O(1)}$ factors respectively.
}

Although Avraham \etal's improvement may seem incremental, this $O^*(n^{6/5})$ bound has remained stagnant (Avraham \etal's paper appeared over 10 years ago), despite more applications and more appearances of this bound in recent papers~\cite{KaplanKSS23,KatzSS24,AgarwalKKS24}; the ideal goal of obtaining near-linear time at the moment seems far out of reach. In this work, we want to understand this class of problems better, and to see if this mysterious exponent $6/5$ can be improved.

In the following, we state the specific problems considered in this class and review their background.

\subparagraph*{Reverse shortest path in (unweighted) unit-disk graphs.}
We start with the simplest of these problems: {\em reverse shortest path} for unit-disk graphs. For every set of points $P$ in the plane and positive real number $r$, let $G_r(P)$ denote the graph with vertex set $P$, where we put an edge between two points iff their Euclidean distance is at most $r$; equivalently, $G_r(P)$ is the intersection graph of the disks with centers $P$ and radius $r/2$.  Note that $G_r(P)$
is a {\em unit-disk graph} (since we can rescale the radius to 1).
In the reverse shortest path problem, one is given a set of points $P$ in the plane, two vertices $s, t \in P$, and a positive integer $\lambda$, and is asked to compute the minimum real number $r^*$ such that $G_{r^*}(P)$ contains an $s$-to-$t$ path with at most $\lambda$ edges. 

This problem is quite natural and can be viewed equivalently as finding the minimum bottleneck path from $s$ to $t$ that
uses at most $\lambda$ links (or ``hops''), where minimizing \emph{bottleneck} refers to minimizing the longest edge in the path.  (As is well known, computing minimum bottleneck paths without constraints on the number of links reduces to minimum spanning trees; {\em e.g.}, see~\cite{Hu61}. It is reasonable to add restrictions on the number of links.)

The decision version of this problem takes an extra parameter $r$, and asks one to determine whether $G_r(P)$ contains an $s$-to-$t$-path with at most $\lambda$ edges. The decision problem (or computing the minimum-link $s$-to-$t$ path in a unit disk graph) can be solved in $O(n \log n)$ time, as shown by Cabello and \Jejcic{}~\cite{CabelloJ15} using the Delaunay triangulation, or alternatively by Chan and Skrepetos~\cite{ChanS16} using a grid-based approach, or 
even earlier by Efrat, Itai, and Katz~\cite{EfratIK01} (who solved a very similar problem as a subroutine for computing bichromatic maximum matchings).
%here, we remark that, although Cabello and \Jejcic{} were the first to state such a result, the problem was actually solved implicitly by Efrat, Itai, and Katz~\cite{EfratIK01} almost a decade earlier, as a subroutine for computing bichromatic maximum matchings. 
Now, since $r^*$ is always the Euclidean distance between two input points, it was immediately noted by Cabello and \Jejcic{} that reverse shortest path can be solved in $\OO(n^{4/3})$ time by binary search, using an $\OO(n^{4/3})$-time distance selection algorithm  (as we have mentioned earlier) to guide the search.

The first algorithm to beat $n^{4/3}$ for reverse shortest path was discovered by Wang and Zhao~\cite{WangZ23}, using a grid-based approach. Their algorithm runs in $\OO(n^{5/4})$ time.  Subsequently, Kaplan \etal~\cite{KaplanKSS23} presented a different, faster randomized algorithm
running in $\OOO(n^{6/5})$ time, by applying the aforementioned
shrink-and-bifurcate technique.

%which was soon afterward subsumed by the $\OOO(n^{6/5})$ result attained by the shrink-and-bifurcate strategy. Despite there being two different approaches, no existing work has combined them to produce a better algorithm.

%Many geometric optimization problems can be solved in roughly the same time bound as their decision versions, using the classical technique of parametric search~\cite{Megiddo83}. Instead of searching over all pairwise distances, roughly speaking, parametric search generates critical values on the fly by simulating the decision algorithm at the unknown $r^*$, performing a binary search only when each independent thread of the decision algorithm encounters a comparison with $r^*$, thus bypassing distance selection. However, parametric search does not perform well when the decision problem has no good parallel algorithm. This is unfortunately the case for reverse shortest path, whose the decision algorithm is an implementation of BFS, which is inherently serial. Since BFS takes $O(n)$ rounds in the worst case, applying standard parametric search yields a running time of $\OO(n^2)$.

%This phenomenon has appeared in several other known geometric optimization problems, as we will list below.

\subparagraph*{Reverse shortest path in weighted unit-disk graphs.} In a \emph{weighted} unit-disk graph, each edge is weighted by the Euclidean distance between the centers of the two disks. The reverse shortest path problem can be defined in a similar way, except that $\lambda$ is now a cap on the sum of the weights along the $s$-to-$t$ path.

Wang and Xue~\cite{WangX20} gave an $O(n\log^2n)$-time algorithm for the decision problem in the weighted case.  Wang and Zhao~\cite{WangZ23} extended their $\OO(n^{5/4})$-time algorithm for reverse shortest path to weighted unit-disk graphs (with a slightly larger polylogarithmic factor).  Kaplan \etal~\cite{KaplanKSS23}
similarly extended their $\OOO(n^{6/5})$-time randomized algorithm to the weighted case, again via shrink-and-bifurcate.

\subparagraph*{Reverse shortest path in unweighted and weighted disk graphs.} The reverse shortest path problem can be generalized to
intersection graphs of disks with arbitrary radii.  The input is a set $S$ of disks in the plane, and $G_r(S)$ is defined as the graph with $S$ as vertices, where we put an edge between two disks iff their distance is at most $r$ (this is sometimes referred to as a \emph{proximity graph}); equivalently, $G_r(S)$
is the intersection graph of the disks with the same centers as $S$ but radii increased additively by $r/2$.
In the weighted case, each edge is weighted by the Euclidean distance between the centers of the two disks (or alternatively, the Euclidean distance of the centers minus the sum of the two radii).

As shown by Kaplan \etal~\cite{KaplanKSS23}, the decision problem can be solved in $\OO(n)$ time by a careful implementation of BFS (in the unweighted case) or Dijkstra's algorithm (in the weighted case), using known dynamic geometric data structures~\cite{KaplanMRSS20} (which have large hidden polylogarithmic factors).  

Wang and Zhao's grid-based approach to
reverse shortest path does not extend to arbitrary disk graphs.
However, the shrink-and-bifurcate approach is still applicable, and yields an
$\OOO(n^{5/4})$-time randomized algorithm as shown by Kaplan \etal~\cite{KaplanKSS23}.

\subparagraph*{Reverse shortest path in segment proximity graphs.}
Another generalization is to the case of line segments.
Given a set $S$ of \emph{disjoint} line segments in the plane, define the segment proximity graph $G_r(S)$ to be the graph with $S$ as vertices, where we put an edge between two segments iff their distance is at most $r$.  The  reverse shortest path problem for (unweighted) line segments can then be defined in the same way.

This problem was first considered by Agarwal \etal~\cite{AgarwalKS24}, who gave a randomized $\OOO(n^{4/3})$-time algorithm.  Later,
Agarwal \etal~\cite{AgarwalKKS24} solved the corresponding decision problem in $O(n\log^2n)$-time.  They then applied the shrink-and-bifurcate technique to obtain an improved randomized algorithm for the optimization problem running in
$\OOO(n^{6/5})$ time, assuming that the spread (the maximum-to-minimum distance ratio) is polynomially bounded.

%\subparagraph*{Reverse shortest path for weighted disk graphs.} 

%Kaplan et al.~\cite{KaplanKSS23} showed how to solve the decision problem by implementing Dijkstra's algorithm using dynamic bichromatic closest pair data structures in $\OO(n)$ time. 

%Similar to BFS, Dijkstra's algorithm is inherently serial, so the standard parametric search technique fails.

%In the special case of weighted unit-disk graphs, {\em i.e.} all disks have the same radius, the decision problem can be solved slightly faster in $O(n \log^2 n)$ time~\cite{WangX20}.

\subparagraph*{Discrete \Frechet{} distance with one-sided shortcuts.} The \emph{discrete \Frechet{} distance} is a popular measure of closeness between two polygonal chains. Given polygonal chains $P=\langle p_1,\ldots,p_{|P|}\rangle$ and $Q=\langle q_1,\ldots,q_{|Q|}\rangle$ with $|P|+|Q|=n$, consider two frogs, a $P$-frog and a $Q$-frog, starting at points $p_1$ and $q_1$ respectively. In each move, one frog leap from its current vertex to the next vertex on their respective polygonal chain, while the other frog stays at its current vertex. The discrete \Frechet{} distance between $P$ and $Q$ is defined as the minimum distance $r^*$ for which the two frogs can arrive at $p_{|P|}$ and $q_{|Q|}$, such that the Euclidean distance between the two frogs is always at most $r^*$ before and after every leap.

The discrete \Frechet{} distance with \emph{one-sided shortcuts} is defined similarly, except that we allow one frog, say the $P$-frog, to leap to any vertex ahead during each move. In other words, the one-sided discrete \Frechet{} distance between $P$ and $Q$ is the minimum discrete \Frechet{} distance between $P'$ and $Q$, over all subsequences $P'$ of $P$ that start at $p_1$ and end at $p_{|P|}$.
The problem of computing the discrete \Frechet{} distance with one-sided shortcuts is what first prompted Avraham \etal~\cite{AvrahamFKKS15} to introduce their shrink-and-bifurcate technique.
They noted a linear-time algorithm for the decision problem and obtained an $O^*(n^{6/5})$-time randomized algorithm for the optimization problem.

\subparagraph*{Reverse shortest path in visibility graphs over 1.5-dimensional terrains.} A \emph{terrain} $T$ in 1.5 dimensions refers to an $x$-monotone polygonal chain in the plane. Two points $p, q$ on or above the terrain are said to {\em see} each other if the line segment $\overline{pq}$ does not intersect $T$. Given a  1.5-dimensional terrain $T$ and a set of points $P$ on or above $T$ with $|P|+|T|=n$, the visibility graph $G(P, T)$ is a graph with vertices $P$ and edges indicating which pairs of points see each other.

In the reverse shortest path problem over a terrain $T$, we are given a set $P$ of points on $T$, and we want to erect a tower at each point in $P$. In particular, given a terrain $T$, a set of points $P$ on $T$, a pair of points $s, t \in P$, and a positive integer $\lambda$, we are asked to find the minimum real number $h^*$ such that the visibility graph $G(P(h^*), T)$ contains an $s$-$t$ path with at most $\lambda$ edges, where $P(h^*)$ is the tips of the towers of height $h^*$ erected at the points of $P$.

This problem was first introduced by Agarwal \etal~\cite{AgarwalKS24}, who gave a randomized $\OOO(n^{4/3})$-time  algorithm.
Katz \etal~\cite{KatzSS24} solved the decision problem in $\OO(n)$ time by a careful implementation of BFS on the visibility graph, with clever uses of range searching data structures.  Katz \etal\ then applied the shrink-and-bifurcate technique to obtain a randomized $\OOO(n^{6/5})$-time algorithm for the optimization problem.

\subsection{The Shrink-and-Bifurcate Technique}

%A breakthrough over the $n^{4/3}$ barrier came in 2015 due to Avraham {\em et al.}~\cite{AvrahamFKKS15}, who improved the running time for discrete \Frechet{} distance with one-sided shortcuts to $\OOO(n^{6/5})$. 

As noted, the current fastest algorithms for all of the above problems are obtained via the shrink-and-bifurcate technique introduced by Avraham \etal~\cite{AvrahamFKKS15}.
%Their work introduced a new technique, which they termed {\em shrink-and-bifurcate}. 
The setting requires that critical values (values that the decision algorithm may compare the input value against) are distances defined by pairs of input objects for some distance function.
The %shrink-and-bifurcate 
strategy consists of two stages, {\em interval shrinking} and {\em bifurcation}.

\begin{enumerate}
    \item  The interval shrinking stage is a relaxed version of binary search---instead of insisting on finding the answer $r^*$ immediately, we seek an interval $(r^-, r^+]$ that contains $r^*$ and that contains at most $L$ critical values for a given parameter $L$. 
In the case of Euclidean distances of points in the plane,
Avraham \etal\ showed that this stage can be done in $\OOO(n^{4/3}/L^{1/3}+D(n))$ time (using randomization), where $D(n)$ is the running time of the decision algorithm.  This is potentially better than the $\OO(n^{4/3})$ time bound for distance selection.

\item
The  bifurcation stage then uses this interval $(r^-, r^+]$ to simulate the decision algorithm at the unknown $r^*$ more efficiently than standard parametric search when the decision algorithm is not parallelizable. Avraham et al.\  showed that this simulation can be done in $\OO(L^{1/2}D(n))$ time.
%where $D(n)$ is the running time of the decision algorithm. 
\end{enumerate}

When $D(n) = \OO(n)$, choosing $L = n^{2/5}$ to balance the costs of the two stages gives an overall running time of $\OOO(n^{6/5})$. 

Some of the applications use distance functions other than Euclidean distances of points in the plane, but the $\OOO(n^{4/3}/L^{1/3}+D(n))$ bound for interval shrinking is still applicable. 
An exception is the reverse shortest path problem for general disk graphs.  Here, the corresponding distance selection problem requires $\OOO(n^{3/2})$ time, and the interval shrinking subproblem is solved in $\OOO(n^{3/2}/L^{1/2}+D(n))$ time instead. 
When $D(n) = \OO(n)$, choosing $L = n^{1/2}$ to balance the costs of the two stages gives an overall running time of $\OOO(n^{5/4})$ instead.

%The shrink-and-bifurcate strategy has since been applied to all the optimization problems aforementioned, and was successful in achieving $\OOO(n^{6/5})$ running time in each of them~\cite{KaplanKSS23,KatzSS24}. However, the running time of $\OOO(n^{6/5})$ has not been improved since it was first achieve in~\cite{AvrahamFKKS15}.

%In the meantime, Wang and Zhao~\cite{WangZ23} independently developed an algorithm for reverse shortest paths for unit-disk graphs, using a grid-based approach that is completely different from the shrink-and-bifurcate strategy. Their algorithm runs in $\OO(n^{5/4})$ time, which was soon afterward subsumed by the $\OOO(n^{6/5})$ result attained by the shrink-and-bifurcate strategy. Despite there being two different approaches, no existing work has combined them to produce a better algorithm.

\subsection{New Results}

The new contributions in this paper are twofold:

\begin{enumerate}
    \item We improve all the previous $\OOO(n^{6/5})$ time bounds---for reverse shortest path in unit-disk graphs, weighted unit-disk graphs, segment proximity graphs, visibility graphs over 1.5-dimensional terrains, and discrete \Frechet{} distance with one-sided shortcuts---to $\OO(n^{8/7})$.
    (In particular, for discrete \Frechet{} distance with one-sided shortcuts, this is the first improvement on the exponent in 10 years.)
    %The improvement, though incremental, brings the exponent noticeably closer to 1 (from 1.2 to about 1.143).

    For reverse shortest path in arbitrary disk graphs (unweighted or weighted), we improve the previous $\OOO(n^{5/4})$ time bound to $\OOO(n^{6/5})$.
    
    Although the improvements are incremental, they are all obtained from one simple idea for improving interval shrinking (we keep the bifurcation part unchanged)---the simplicity and generality of our idea make it easy to apply to existing (and potentially future) applications of the shrink-and-bifurcate approach, and are the main takeaways of the paper. More precisely, we show how to lower the $\OOO(n^{4/3}/L^{1/3}+D(n))$ time bound for interval shrinking subproblem to $\OO(n^{4/3}/L^{2/3}+D(n))$ (or in the arbitrary disk graph case, lower $\OOO(n^{3/2}/L^{1/2}+D(n))$ to $\OOO(n^{3/2}/L^{3/4}+D(n))$).  The original interval shrinking method by Avraham \etal\ needs to modify range searching techniques for distance selection, constructing ``partial biclique covers'' and then applying random sampling. Our new method, in contrast, can be described in just a few lines, and simply uses a known distance selection algorithm \emph{as a black box} on random subsets (see Section~\ref{sec:interval-shrinking}).

   \item Next, we focus more closely on the simplest of this class of problems, namely, reverse shortest path in (unweighted) unit-disk graphs. For this specific problem, we are able to further improve the running time to $\OO(n^{9/8})$. The improvement of the exponent from $6/5=1.2$ to $9/8=1.125$ is a bigger increment (for example, comparable in magnitude to Wang and Zhao's original improvement~\cite{WangZ23} from $4/3\approx 1.333$ to $5/4=1.25$, or 
   Kaplan \etal's improvement~\cite{KaplanKSS23} from $5/4=1.25$ to $6/5=1.2$).

    Interestingly, we obtain our result by combining the shrink-and-bifurcate technique with Wang and Zhao's original grid-based approach!  Although it may be natural to consider combining the two approaches, how to do so is not obvious at all.  We need to put together several nontrivial technical ideas---this is the most intricate part of the paper (see Section~\ref{sec:rsp-uudg}). 
\end{enumerate}

%In this work, we improve the interval shrinking procedure of \cite{AvrahamFKKS15} for the first time, improving the running time from $\OOO(n^{4/3}/L^{1/3})$ to $\OOO(n^{4/3}/L^{2/3})$. This immediately allows us to improve the running for all problems from $\OOO(n^{6/5})$ to $\OOO(n^{8/7})$. Furthermore, our interval shrinking algorithm is simple. In particular, the original interval shrinking algorithm opens the blackbox of known range counting data structures to construct a partial biclique cover. Our algorithm, in contrast, can be described in just a few lines, using known distance selection algorithms as blackbox.

%Furthermore, for the special case of reverse shortest path in unweighted unit-disk graphs, we manage to combine the shrink-and-bifurcate technique with Wang and Zhao's grid-based technique, improving the running time further to $\OOO(n^{9/8})$. Combining the two approaches turns out not to be straightforward, and requires a nontrivial new ingredient.

%%%%%%%%%%%%%%%%%%%%%%%%%%%%%%%%%%%%%%%%%%%%%%%%%%%%%%%%%%%%%%%%%%%%%%
%%%                Interval Shrinking & Bifurcation                %%%
%%%%%%%%%%%%%%%%%%%%%%%%%%%%%%%%%%%%%%%%%%%%%%%%%%%%%%%%%%%%%%%%%%%%%%

\section{Interval Shrinking}
\label{sec:interval-shrinking}

In this section, we formally define the {\em interval shrinking} subproblem in an abstract, self-contained setting, and then describe our new randomized method for this subproblem.

\begin{subproblem}[Interval shrinking]
    \label{prb:interval-shrinking}
    Suppose we have an unknown value $r^*$. Suppose there is a \emph{decision oracle} that can decide whether $r^*\le r$ for an input value $r$ in $D(n)$ time where $D(n)\ge n$.

    Let $\kkappa(\cdot,\cdot)$ be a  function over pairs of objects.  Suppose there is a \emph{selection oracle} that, given two sets $P$ and $Q$ of objects and a number $k$, can compute the $k$-th smallest value in $\kkappa(P, Q) := \{ \kkappa(p, q) \mid (p, q) \in P \times Q \}$ in $K(|P|,|Q|)=\OO(|P|^{\alpha}|Q|^{\alpha}+|P|+|Q|)$ time for some constant $\alpha < 1$.
%Suppose that, given two sets of points $P$ and $Q$, one can find the predecessor and successor of $r^*$ in $\kkappa(P, Q) := \{ \kkappa(p, q) \mid (p, q) \in P \times Q \}$ in $K(m, n)$ time.
    
    Given two sets $P$ and $Q$ of objects with $|P|+|Q|=n$ and a number $L$, we want to compute an interval $(r^-, r^+]$ that contains $r^*$ and that contains at most $\OO(L)$ values in $\kkappa(P, Q)$.
\end{subproblem}

If $\kkappa$ is the Euclidean distance function for points in the plane (which is the case for the applications to reverse shortest path in unit-disk graphs or discrete \Frechet{} distance with one-sided shortcuts),
known distance selection algorithms~\cite{AgarwalASS93,KatzS97,WangZ23Improved,ChanZ24} 
between an $m$-point set and an $n$-point set run in
$K(m,n)=\OO(m^{2/3} n^{2/3} + m + n)$ time, and thus we have $\alpha=2/3$.  In this planar Euclidean case, Avraham \etal~\cite{AvrahamFKKS15} solved
the above interval shrinking problem in $\OOO(n^{4/3}/L^{1/3}+D(n))$ time with high probability.  We will describe an improved method.

It suffices to describe how to find an $r^+>r^*$ such that
$[r^*,r^+]$ contains at most $\OO(L)$ values in $\kkappa(P,Q)$, since finding $r^-$ is analogous.
Also, it suffices to design an algorithm which is correct with probability $\Omega(1)$, since we can get an algorithm which is correct w.h.p.\footnote{
 With high probability, i.e., probability $1-O(1/n^c)$ for an arbitrarily large constant $c$.
 }\ 
 by repeating for logarithmically many trials and returning the minimum $r^+$ found.

%All problems considered in this paper satisfy $K(m, n) = \OOO(m^{2/3} n^{2/3} + m + n)$. For problems where $\kkappa(p, q)$ is the Euclidean distance $\|pq\|$, given any two sets of points $P$ and $Q$, using distance selection~\cite{KatzS97} as a blackbox, we can compute the predecessor and successor of $r^*$ in $\|PQ\| := \{ \|pq\| \mid (p, q) \in P \times Q \}$ in $\OO(m^{2/3} n^{2/3} + m + n)$ time. The case of reverse shortest path in visibility graphs in 1.5-dimensional terrains will be discussed in Appendix~\ref{sec:terrain}.

\subsection{First Version}

Our first interval shrinking algorithm is simple, as presented in Algorithm~\ref{alg:interval-shrinking-1}.

\begin{algorithm}
    \caption{Basic interval shrinking\vspace{0.2em}}
    \label{alg:interval-shrinking-1}
    Take random subsets $\widehat{P} \subseteq P$ and $\widehat{Q} \subseteq Q$ where each element is chosen independently with probability $\frac1{\sqrt{L}}$

    \Return the successor $r^+$ of $r^*$ in $\kkappa(\widehat{P},Q)\cup \kkappa(P,\widehat{Q})$,  which can be computed by binary searches in $\kkappa(\widehat{P},Q)$ and $\kkappa(P,\widehat{Q})$ using the selection and decision oracles
\end{algorithm}

\subparagraph*{Correctness.} 
Let $E \subseteq P \times Q$ be the set of $L$ pairs that have the $L$ smallest $\kkappa(\cdot,\cdot)$ values after $r^*$.
%arbitrary set of $L$ pairs. 
Think of $E$ as a bipartite graph between $P$ and $Q$. Then $E$ has at least $\sqrt{L}$ non-isolated vertices in $P$ or at least $\sqrt{L}$ non-isolated vertices in $Q$ (because otherwise, $E$ would have fewer than $L$ edges). W.l.o.g., assume the former.  Then $\widehat{P} \times Q$  hits $E$  with probability at least $1-(1-\frac1{\sqrt{L}})^{\sqrt{L}}=\Omega(1)$.  This implies that there are at most $L$ values in $\kkappa(\widehat{P},Q)$ between $r^*$ and $r^+$ with probability $\Omega(1)$.

%In the second case, $P \times \widehat{Q}$ hits $E$ with constant probability.

\subparagraph*{Running time.} 

Note that $|\widehat{P}|,|\widehat{Q}|=\OO(\frac{n}{\sqrt{L}})$ w.h.p.
Using the selection and decision oracles, w.h.p., the binary searches can be done in time
\begin{align*}
    \OO\left(
        \left(\frac{n}{\sqrt{L}}\right)^{\alpha} n^{\alpha} \:+\: n^\alpha \left(\frac{n}{\sqrt{L}}\right)^{\alpha} \:+\:  D(n)
    \right)
    \ =\ \OO(n^{2\alpha}/L^{\alpha/2} + D(n)).
\end{align*}

For planar Euclidean distances with $\alpha=2/3$, 
we have thus recovered the $\OOO(n^{4/3}/L^{1/3}+D(n))$ result of Avraham {\em et al.}~\cite{AvrahamFKKS15}. However, our algorithm is arguably simpler, because, unlike Avraham {\em et al.}, we invoke distance selection as a black box, without needing to modify the internal mechanisms of distance selection. %% TODO: remark the subpolynomial speedup.

We remark that the idea of randomly sampling \emph{pairs} of input objects is more common and has been used before in the context of slope selection or distance selection~\cite{Matousek91a,DillencourtMN92,Chan01} (and in Avraham \etal's previous method~\cite{AvrahamFKKS15}).  But here we are sampling the input objects themselves (which in some sense is even simpler, though the analysis is slightly trickier).

\subsection{Improved Version}

Upon closer examination, one may notice that the basic algorithm has room for improvement. In particular, if  $E$ has $\Theta(\sqrt{L})$ non-isolated vertices in $P$ and $\Theta(\sqrt{L})$ non-isolated vertices in $Q$, then clearly our samples $\widehat{P} \times Q$ and $P \times \widehat{Q}$ are wasteful. Indeed, we can eliminate such redundancy by sampling in $P$ and $Q$ simultaneously with different rates. We present the details in Algorithm~\ref{alg:interval-shrinking-2}.

\begin{algorithm}
    \caption{Improved interval shrinking\vspace{0.2em}}
    \label{alg:interval-shrinking-2}

    \For{$i = 1, \ldots, \log n$}
    {
        Take a random subset $\widehat{P}_i \subseteq P$
               where each element is chosen independently with probability $\frac1{\ceil{L/2^i}}$ \\
        Take a random subset $\widehat{Q}_i \subseteq Q$
               where each element is chosen independently with probability $\frac1{2^i}$ \\
    }
    \Return the successor $r^+$ of $r^*$ in $\bigcup_{i=1}^{\log n} \kkappa(\widehat{P}_i,\widehat{Q}_i)$, which can be computed by binary search in each $\kkappa(\widehat{P}_i,\widehat{Q}_i)$ using the selection and decision oracles
    
\end{algorithm}

\subparagraph*{Correctness.} 
Let $E \subseteq P \times Q$ be the set of $L\log n$ pairs that have the $L\log n$ smallest $\kkappa(\cdot,\cdot)$ values after $r^*$.
%Let $E \subseteq P \times Q$ be an arbitrary set of $L$ pairs of points in $P$ and $Q$. 
Consider the bipartite graph $H = (P \sqcup Q, E)$. Let $M_i \subseteq P$ be the vertices with degree in the range $[2^{i-1}, 2^{i})$ in $H$. Since $\sum_{i=1}^{\log n} 2^i |M_i|\ge L\log n$, there must exist an index $i_0$ for which $|M_{i_0}|\ge \ceil{L/2^{i_0}}$. Observe that $\widehat{P}_{i_0} \times \widehat{Q}_{i_0}$ hits $E$ if 
\begin{enumerate}
    \item[(i)] $\widehat{P}_{i_0}$ hits $M_{i_0}$, and
    \item[(ii)] for some $p\in \widehat{P}_{i_0}\cap M_{i_0}$,
    $\widehat{Q}_{i_0}$ hits the neighborhood of $p$ in $E$.
\end{enumerate}

Now, (i) holds with probability at least $1-(1-\frac1{\ceil{L/2^{i_0}}})^{\ceil{L/2^{i_0}}}=\Omega(1)$.  For a fixed $p\in \widehat{P}_{i_0}\cap M_{i_0}$, the probability that
    $\widehat{Q}_{i_0}$ contains at least one neighbor of $p$ is at least $1-(1-\frac1{2^{i_0}})^{2^{i_0-1}}=\Omega(1)$.
Thus, conditioned on (i), the probability of (ii) is $\Omega(1)$.
Hence,  $\widehat{P}_{i_0} \times \widehat{Q}_{i_0}$ hits $E$ with probability $\Omega(1)$. This implies that there are at most $L\log n$ values in $\kkappa(P,Q)$ between $r^*$ and $r^+$ with probability $\Omega(1)$.

%$ \cap \delta_H(p) \neq \emptyset$, both of which happen with constant probability.

\subparagraph*{Running time.} 
Note that $|\widehat{P}_i|=\OO(\frac{n}{\ceil{L/2^i}})$
and $|\widehat{Q}_i|=\OO(\frac{n}{2^i})$ w.h.p.
Using the selection and decision oracles, w.h.p.\ the binary searches can be done in time
\begin{align*}
    \OO\left(\sum_{i=1}^{\log n}
       \left( \left(\frac{n}{\ceil{L/2^i}}\right)^{\alpha}
        \left(\frac{n}{2^i}\right)^{\alpha} \:+\: D(n) \right)
    \right)
    \ =\ \OO(n^{2\alpha}/L^{\alpha} + D(n)).
\end{align*}

%The running time analysis is similar to that for the basic algorithm. For each $i = O(\log n)$, we have $|\widehat{P}_i| = O(2^i|P|/L)$ and $|\widehat{Q}_i| = O(|Q|/2^i)$. Thus, we can perform binary search on $\| \widehat{P}_i \widehat{Q}_i \|$ in $\OO(m^{2/3}n^{2/3} / L^{2/3} + m + n)$ time.

\begin{theorem}
    \label{thm:interval-shrinking}
    Subproblem~\ref{prb:interval-shrinking} (interval shrinking) can be solved in $\OO(n^{2\alpha}/L^{\alpha} + D(n))$ time w.h.p.
\end{theorem}

For example, for planar Euclidean distances, the bound is 
$\OO(n^{4/3}/L^{2/3} + D(n))$, which is noticeably better than $\OO(n^{4/3}/L^{1/3} + D(n))$.

\subsection{Combining with Bifurcation}

After interval shrinking, Avraham \etal~\cite{AvrahamFKKS15} devised a clever variant of parametric search via an idea called \emph{bifurcation} (also redescribed in more generality in Kaplan \etal's subsequent paper~\cite{KaplanKSS23}) which accomplishes the following:

%Bifurcation is a procedure that allows one to simulate any algorithm at $r^*$, provided that 1) we have an interval $(\alpha, \beta]$ containing at most $L$ critical values and 2) there is an oracle that can compare any real number $r$ with $r^*$. In this paper, we use bifurcation as a blackbox, as stated in the following lemma.

\begin{lemma}[Bifurcation \cite{AvrahamFKKS15}]\label{lem:bifurcate}
    Suppose we have an unknown value $r^*$.  Suppose there is a \emph{decision oracle} that can decide whether $r^*\le r$ for an input value $r$ in $D(n)$ time.

    Suppose we have an algorithm $\mathcal{A}$ that has one input parameter and runs in no more than $O(D(n))$ time, and during the execution of the algorithm, it can only access its input parameter via binary comparisons. The values that $\mathcal{A}$ can compare its input parameter with are called the {\em critical values}.

    Lastly, suppose that we are also given an interval $(r^-, r^+]$ that is guaranteed to contain $r^*$ and that contains at most $\OO(L)$ critical values.
    
    Then we can simulate $\mathcal{A}$ on input $r^*$ in $\OO(L^{1/2} D(n))$ time.
\end{lemma}

Typically, in applications of the technique, the algorithm $\mathcal{A}$ we want to simulate is the decision algorithm itself; here, the simulation produces an interval containing $r^*$ such that all input values lying in the interval have the same computation path and thus the same answer as $r^*$  (namely, ``yes''), and so the left endpoint of this interval gives us precisely the value of $r^*$.

We will not redescribe the proof of the above lemma, as we will only use it as a black box. (The rough idea in bifurcation is to simulate multiple branches of $\mathcal{A}$ at the same time, until we have accumulated sufficiently many comparisons and will then resolve all of them as a batch using the decision oracle; comparisons with values outside the interval $(r^-,r^+]$ are free.)   Combining our improved interval shrinking method with bifurcation, we immediately obtain better results for all optimization problems stated in the introduction:

\begin{theorem}\label{thm:appl}
\
\begin{enumerate}
\item[(a)]
    The reverse shortest path problem in unweighted and weighted unit-disk graphs, segment proximity graphs (assuming polynomial spread), and visibility graphs over 1.5-dimensional terrains, and the discrete \Frechet{} distance problem with one-sided shortcuts, as defined in the introduction, can all be solved 
    in $\OO(n^{8/7})$ time w.h.p.\ by (Las Vegas) randomized algorithms.
\item[(b)]
    The reverse shortest path problem in unweighted and weighted disk graphs can be solved in $\OOO(n^{6/5})$ time w.h.p.\ by (Las Vegas) randomized algorithms.
\end{enumerate}
\end{theorem}
\begin{proof}
   As mentioned in the introduction, all these problems have known decision algorithms running in $D(n)=\OO(n)$ time \cite{AvrahamFKKS15,KaplanKSS23,KatzSS24,AgarwalKKS24}.  In all these decision algorithms, critical values are indeed obtained from some function $\kkappa$ over pairs of input objects.

   For all the applications in (a), the function $\kkappa$ satisfies $\alpha=2/3$.
  (Some care is needed for reverse shortest path in visibility graphs in 1.5-dimensional terrains; see Appendix~\ref{sec:terrain}.)
  
   Combining Theorem~\ref{thm:interval-shrinking} with
   Lemma~\ref{lem:bifurcate}, we get total running time
   \[ \OO\left( \frac{n^{4/3}}{L^{1/3}} + L^{1/2} n\right). \]
   Choosing $L=n^{2/7}$ yields $\OO(n^{8/7})$.
     
   For the applications in (b), we have $\alpha=3/4+o(1)$~\cite{KaplanKSS23}.  The total running time becomes 
   \[ \OOO\left( \frac{n^{3/2}}{L^{3/4}} + L^{1/2} n\right). \]
   Choosing $L=n^{2/5}$ yields $\OOO(n^{6/5})$.
\end{proof}

\section{Reverse Shortest Path in Unweighted Unit-Disk Graphs}
\label{sec:rsp-uudg}

For reverse shortest path in unweighted unit-disk graphs, we  show how to improve the running time even further. 

In this section, $\kkappa(p,q)$ denotes the Euclidean
distance between points $p$ and $q$, and
$\kkappa(P,Q) := \{\kkappa(p,q)\mid (p,q)\in P\times Q\}$ for two point sets $P$ and $Q$.  Let $d_G(s,t)$ denote the shortest-path distance between $s$ and $t$ in a graph~$G$.

We are given an input set $P$ of $n$ points in the plane, $s,t\in P$, and a number $\lambda$.  Let $G_r(P)$ be the (unweighted) graph with vertex set $P$, where we put an edge between $p$ and $q$ iff $\kkappa(p,q)\le r$.
We want to compute the smallest $r^*$ for which
$d_{r^*}(s,t)\le \lambda$, where we use $d_r(s,t)$ as a shorthand for $d_{G_r(P)}(s,t)$.

\subsection{Recap of Chan and Skrepetos' Decision Algorithm}

Our algorithms rely on the details of the known $\OO(n)$-time decision algorithm by Chan and Skrepetos~\cite{ChanS16}, so we will first give a quick recap of their algorithm.

Given an input value $r$, we want to compute $d_r(s,t)$.
We first create a uniform grid $\Psi_r$ of side length $r/\sqrt{2}$.
Chan and Skrepetos' algorithm is a careful implementation of BFS using the grid.
We start from $s$ and generate each level of the BFS tree (the ``frontier'') one at a time.
The key observation is that each grid cell only needs to be visited a constant number of times during the BFS, because (i)~points in the same cell
have the roughly same shortest-path distance $\pm 1$ from $s$ in $G_r(P)$ and so each cell ``appears'' in  only $O(1)$ levels of the BFS tree, and (ii)~each cell
has only $O(1)$ neighboring cells---we say that two cells are \emph{neighboring} if the minimum Euclidean distance between the two cells is at most $r$.
Given one level of the BFS tree, we can generate the next level (i.e., ``expand'' the frontier) by
solving the following subproblems between cells in the current level and their neighboring cells:

%and then implement BFS based on the grid. In particular, we maintain a set of frontier points $Q$. In each BFS round, for each cell $\sigma$ that contains a point from $Q$, we expand the frontier by solving the following problem for each $\sigma'$ within distance $r$ from $\sigma$: report all points $p \in P\cap \sigma'$ such that $\kkappa(p,q) \leq r$ for some $q \in Q\cap \sigma$. When $\sigma' = \sigma$, the problem is trivial. Otherwise, we call the problem a {\em frontier problem in $(\sigma, \sigma')$}, and it can be phrased as follows.

\begin{subproblem}
    \label{prb:frontier}
    Given an input value $r$, a subset of $n_r$ red points in one grid cell $\sigma$ and a subset of $n_b$ blue points in a neighboring grid cell $\sigma'$, determine for each blue point whether there is a red point at distance at most $r$ from it.
\end{subproblem}

Chan and Skrepetos solved the above subproblem in $O(n_r + n_b)$ time, assuming that the input points are pre-sorted, by using envelopes of pseudo-lines.
As a result, the total running time of their algorithm is linear after pre-sorting.

The above algorithm is not efficiently parallelizable (since the number of levels in the BFS may be large).  Instead, we can apply the bifurcation technique to simulate the algorithm on~$r^*$.  However, one issue arises: the envelope computations requires critical values defined by triples of input elements, not pairs.
As noted in other papers~\cite{KaplanKSS23,WangZ23},
this can be fixed by switching to a different solution to Subproblem~\ref{prb:frontier} which is slightly slower by a logarithmic factor but is simpler: namely, we compute the nearest red neighbor to each blue point by point location in the Voronoi diagram of the red points in $O((n_r+n_b)\log (n_r+n_b))$ time~\cite{BergCKO08};
then for each blue point, we just compare $r$ with its nearest neighbor distance.  This way, the critical values are all Euclidean distances of pairs of input points.

Another issue is the construction of the grid $\Psi_r$, i.e., the assignment of input points to grid cells in $\Psi_r$.  This part is easily parallelizable, and so we can apply standard parametric search to construct the grid $\Psi_{r^*}$ in $\OO(n)$ time, all done as preprocessing.  (In fact, the time bound for constructing $\Psi_{r^*}$ is $O(n\log n)$, as shown by Wang and Zhao~\cite{WangZ23}.)

%Using a static range emptiness data structure, the frontier problem in $(\sigma, \sigma')$ can be solved in $O((n_r + n_b) \log (n_r + n_b))$ time. If in addition all points have been sorted in advance, then this problem can be solved in $O(n_r + n_b)$ time with a Graham's scan. Since a cell $\sigma$ contains frontier points in at most two rounds, the algorithm of Chan and Skrepetos takes $O(n \log n)$ time.

%Based on Chan and Skrepetos' algorithm, Wang and Zhao~\cite{WangZ23} proved two results that will be helpful for our reverse shortest path algorithm.

%\begin{lemma}[Wang and Zhao~\cite{WangZ23}]
%    \label{lem:wz-grid}
%    In $O(n \log n)$ time, one can construct a grid $\Psi$ with the same combinatorial structure as the grid $\Psi_{r^*}$.
%\end{lemma}

%\begin{lemma}[Wang and Zhao~\cite{WangZ23}]
%    \label{lem:wz-round}
 %   Let $(R_1, B_1, r^*), \ldots, (R_\ell, B_\ell, r^*)$ be a collection of instances of Subproblem~\ref{prb:frontier} at the unknown $r^*$ with total size $n$. Then one can solve all instances in $O(n \log n)$ time.
%\end{lemma}

\subsection{Heavy vs.\ Light Cells}

Applying the shrink-and-bifurcate technique to the preceding decision algorithm using our improved interval shrinking method yields an overall time bound of $\OO(n^{8/7})$, as we have already shown in Theorem~\ref{thm:appl}(a).
To obtain our best result, we will combine with another idea that was used in
Wang and Zhao's original paper~\cite{WangZ23}: heavy vs.\ light cells.

Recall that we have already constructed the grid $\Psi_{r^*}$.
%is the grid that would have been constructed by Chan and Skrepetos' decision algorithm at $r^*$. Let $\Psi$ be the grid constructed by Lemma~\ref{lem:wz-grid}. In any grid, we 
Let $\Delta$ be a parameter to be chosen later.
Call a grid cell $\sigma$ {\em heavy} if $|P \cap \sigma| \geq \Delta$, or {\em light} otherwise. The number of heavy cells is clearly at most $O(n/\Delta)$. Furthermore, a pair of neighboring cells $(\sigma, \sigma')$ is called a {\em light pair} if one of $\sigma$ and $\sigma'$ is a light cell (the other may be light or heavy).

Our new strategy is as follows. First, we perform interval shrinking, but only with respect to distances of points from light pairs.
%of $\Psi$, {\em i.e.,}
%\[
%\bigcup_{\substack{\textrm{$(\sigma, \sigma')$ is}\\ \textrm{a light pair}}}
%        \kkappa(P\cap\sigma, P\cap \sigma').
%\]
Intuitively, there are fewer such distances and so interval shrinking should be doable faster.  This allows us to simulate Chan and Skrepetos' algorithm at $r^*$ on a modified graph where heavy cells are eliminated---or more precisely, points in a common heavy cell are contracted into a single vertex.
When we ``uncontract'' the heavy cells, this causes an $O(n/\Delta)$ additive error to all the shortest-path distances in $G_{r^*}(P)$.
As a final step, we show how to use a dynamic program to recover the value $r^*$.
In the next three subsections, we provide the details of all these steps.

\subsection{Faster Interval Shrinking for Light Pairs}

\newcommand{\dlight}{\delta_{\textrm{light}}}

In the following, we show how to solve the interval shrinking problem faster for distances from light pairs.  More precisely, 
define $\dlight(p,q)=\kkappa(p,q)$ if $p\in\sigma$ and $q\in\sigma'$ for some light pair $(\sigma,\sigma')$,
and $\dlight(p,q)=\infty$ otherwise.
We apply Algorithm~\ref{alg:interval-shrinking-2} to do interval shrinking with respect to $\dlight$.

To analyze the running time, we need an efficient implementation of the selection oracle for 
\[ \dlight(\widehat{P}_i,\widehat{Q}_i)\ = 
           \bigcup_{\textrm{light pair }(\sigma, \sigma')}
            \kkappa(\widehat{P}_i\cap \sigma, \widehat{Q}_i\cap \sigma') \ \cup\ \{\infty\}.
\]

%we want to compute an interval $(r^-, r^+]$ that contains $r^*$ and that contains at most $\OO(L)$ values in 
%\[
%\bigcup_{\textrm{light pair }(\sigma, \sigma')} \kkappa(P\cap\sigma, P\cap \sigma').
%\]
%As before, it suffices to describe how to compute $r^+$ with probability $\Omega(1)$.

%The details of the algorithm are similar to Algorithm~\ref{alg:interval-shrinking-2} and are presented in Algorithm~\ref{alg:interval-shrinking-light}.

%\begin{algorithm}
%    \caption{Improved interval shrinking for light pairs\vspace{0.2em}}
%    \label{alg:interval-shrinking-light}

%    \For{$i = 1, \ldots, \log n$}
%    {
%        Sample $\widehat{P}_i \subseteq P$
%               with probability $\frac1{\ceil{L/2^i}}$
%        and    $\widehat{Q}_i \subseteq P$
%               with probability $\frac1{2^i}$ \\
%    }
%    \Return the successor $r^+$ of $r^*$ in
%    $\displaystyle
%        \bigcup_{i=1}^{\log n}
%        \bigcup_{\textrm{light pair }(\sigma, \sigma')}
%            \kkappa(\widehat{P}_i\cap \sigma, \widehat{Q}_i\cap \sigma')
%    $, which can be computed by binary search using a selection oracle and the decision oracle
%\end{algorithm}

%\subparagraph*{Correctness.}  Same as before.

%\subparagraph*{Running time.}
Let $n_\sigma := |P\cap \sigma|$ for each cell $\sigma$.
Note that $|\widehat{P}_i\cap\sigma|=\OO(\frac{n_\sigma}{\ceil{L/2^i}})$
and $|\widehat{Q}_i\cap\sigma'|=\OO(\frac{n_{\sigma'}}{2^i})$ for every $(\sigma,\sigma')$ w.h.p.
We already have a selection oracle for each $\delta(\widehat{P}_i\cap\sigma,\widehat{Q}_i\cap\sigma')$
that w.h.p.\ runs in time
\[
\OO\left( \left(\frac{n_\sigma}{\ceil{L/2^i}}\right)^{2/3}
            \left(\frac{n_{\sigma'}}{2^i}\right)^{2/3} + n_\sigma+n_{\sigma'}\right)
            \ =\ \OO\left(\frac{n_\sigma^{2/3}n_{\sigma'}^{2/3}}{L^{2/3}} + n_\sigma+n_{\sigma'}\right).
\]
%However, we still need a selection oracle for the union $\displaystyle\bigcup_{\textrm{light pair }(\sigma, \sigma')}
%            \kkappa(\widehat{P}_i\cap \sigma, \widehat{Q}_i\cap \sigma')$ in order to implement the above binary search.

To obtain a selection oracle for $\dlight(\widehat{P}_i,\widehat{Q}_i)$,
we face a problem analogous to selection in a union of multiple sorted arrays, except that the arrays are implicitly given.
Frederickson and Johnson~\cite{FredericksonJ82} gave an algorithm for selection in $k$ sorted arrays of size $n$ running in $O(k\log n)$ time; the algorithm proceeds in $O(\log n)$ rounds, and in each round, it looks up only $O(1)$ elements in each array.  In our scenario, each array lookup corresponds to a call to the selection oracle for $\delta(\widehat{P}_i\cap\sigma,\widehat{Q}_i\cap\sigma')$ for a light pair $(\sigma,\sigma')$.  Thus, by Frederickson and Johnson's algorithm, selection in the union $\dlight(\widehat{P}_i,\widehat{Q}_i)$ can be done in total time
\begin{eqnarray*} 
\lefteqn{\OO\left( \sum_{\textrm{light pair }(\sigma, \sigma')} \left(\frac{n_\sigma^{2/3}n_{\sigma'}^{2/3}}{L^{2/3}} + n_\sigma+n_{\sigma'}\right)  \right)}\\
&\le& \OO\left( \sum_{\textrm{light pair }(\sigma, \sigma')} \frac{\Delta^{1/3}(n_\sigma+n_{\sigma'})}{L^{2/3}} + n_\sigma+n_{\sigma'}  \right)
\ =\ \OO\left(\frac{\Delta^{1/3} n}{L^{2/3}} + n \right),
\end{eqnarray*}
since  $\min\{n_\sigma,n_{\sigma'}\}\le \Delta$ for each light pair $(\sigma,\sigma')$, and
$\sum_{\textrm{neighboring pair }(\sigma,\sigma')} (n_\sigma+n_{\sigma'})=O(n)$.
Therefore, the total time for interval shrinking with respect to $\dlight$
is $\OO(\Delta^{1/3}n/L^{2/3}+D(n))$ w.h.p.

\subsection{Computing Distances with $O(n/\Delta)$ Additive Error}

%For every positive real number $r$, we use $G_r$ to denote the unit-disk graph on $P$ with radius $r$. 
Let $H_r$ be the graph obtained from $G_r(P)$ by contracting the points inside each heavy cell of $\Psi_r$ into a single vertex. 
%We refer to edges incident to a contracted vertex as {\em contracted edges}. 
Since there are $O(n/\Delta)$ contracted vertices in $H_r$, we have $d_{r}(s, p) - k  \leq d_{H_r}(s, p) \leq d_{r}(s, p)$ with $k=O(n/\Delta)$, for every $p\in P$.

It is straightforward to modify Chan and Skrepetos' decision algorithm to do
BFS on this modified graph $H_r$.  Expanding the frontier again reduces to
solving Subproblem~\ref{prb:frontier} for pairs of neighboring cells.
Note that for two neighboring cells that are both heavy, we can determine
whether there is an edge between them by pre-computing the bichromatic closest pair between the points in the neighboring cells (via Voronoi diagrams) and testing $r$ against this closest pair distance---call this a \emph{special} distance.
This pre-computation takes $O(n\log n)$ total time.

After performing interval shrinking as described in previous subsection,
we apply the bifurcation technique (Lemma~\ref{lem:bifurcate}) to simulate BFS on $H_{r^*}$.  The critical values indeed are all Euclidean distances from light pairs only, except for the $O(n)$ special distances above,
but initially before the simulation starts, we can resolve the comparisons with all $O(n)$ special distances by binary search using $O(\log n)$ calls to the decision algorithm in $\OO(n)$ time. 

At the end of the simulation, we will then have computed $d_{H_{r^*}}(s,p)$ for all $p\in P$, which yield estimates to $d_{r^*}(s,p)$ with additive error $O(n/\Delta)$. To summarize:

%For any given $r$, BFS on $H_r$ can be done using a straightforward modification to Chan and Skrepetos' algorithm. Indeed, in each iteration, to update the frontier, we 1) solve an instance of Subproblem~\ref{prb:frontier} in each pair of nearby light cells incident to the frontier, and 2) simply explore each contracted edge incident to the frontier.

%\begin{lemma}
%    The contracted vertices and contracted edges of $H_{r^*}$ at the unknown $r^*$ can be constructed in $O(n \log n)$ time.
%\end{lemma}

%\begin{proof}
%    The vertices of $H_{r^*}$ at the unknown $r^*$ can be constructed easily---simply contract vertices inside each heavy cell of $\Psi$ into a single vertex. Computing the contracted edges of $H_{r^*}$ reduces to solving an instance $(P[\sigma], P[\sigma'], r^*)$ of Subproblem~\ref{prb:frontier} for each heavy cell $\sigma$ and each cell $\sigma' \in N_{5\times 5}(\sigma)$. The instances have total size $O(n)$. Thus, by Lemma~\ref{lem:wz-round}, the contracted edges of $H_{r^*}$ can be constructed in $O(n \log n)$ time.
%\end{proof}

%Since the contracted edges at the unknown $r^*$ are computed in advance, the only critical values for simulating BFS on $H_{r^*}$ are the distances $\|pq\|$ where $p$ and $q$ belong to cells that form a light pair. Thus, after the interval shrinking procedure of Algorithm~\ref{alg:interval-shrinking-light}, we can simulate BFS on $H_{r^*}$, using bifurcation, in $O(L^{1/2} n \log n)$ time. This gives us the following intermediate result.

\begin{lemma}
    \label{lem:approx-search}
    In time $\OO(\Delta^{1/3} n / L^{2/3} + L^{1/2} n)$, one can compute for every $p \in P$ an estimate $\dd_p$ with $d_{r^*}(s, p) - k \leq \dd_p \leq d_{r^*}(s, p)$, where $k = O(n / \Delta)$.
\end{lemma}

\subsection{Recovering $r^*$}

Finally, we will compute $r^*$ using the estimates obtained in Lemma~\ref{lem:approx-search}. This final step is done via a dynamic program, and does not require any knowledge of the details of the earlier steps.

\begin{lemma}
    \label{lem:correction}
 Given a number $k$ and estimates $\dd_p$ with $d_{r^*}(s, p) - k \leq \dd_p \leq d_{r^*}(s, p)$ for all $p \in P$, one can compute $r^*$ in $\OO(kn)$ time.
\end{lemma}
% and $\dd_s = 0$

To prove Lemma~\ref{lem:correction}, it is helpful to view the problem
as computing a \emph{minimum bottleneck path} from $s$ to $t$ with at most $\lambda$ links (as mentioned in the introduction),
i.e., a path from $s$ to $t$ with at most $\lambda$ links that minimizes the maximum Euclidean
length of the edges along the path.

There is a straightforward dynamic programming solution to this problem:
Let $D_i(p)$ denote the minimum bottleneck value over all paths from $s$ to $p$ with $i$ links.
We want $r^*=\min_{i\le\lambda} D_i(t)$.
We have the following recursive formula: for each $p\in P$,
\begin{align}
    \label{eqn:correction}
    D_i(p) = \min_{q\in P} \max \{ D_{i-1}(q), \kkappa(p,q) \},
\end{align}
with $D_0(s)=0$ and $D_0(p)=\infty$ for all $p\in P-\{s\}$ as base cases.
Unfortunately, in this dynamic program, the number of table entries is $O(n\lambda)$, which is $O(n^2)$ in the worst case.

We use the given estimates $\dd_p$ to speed up the dynamic program.
The main observation is that it suffices to generate $D_i(p)$ when  $i-k\le\dd_p\le i$.
The number of such table entries is thus reduced to $O(kn)$.

\newcommand{\Dhat}{\widehat{D}}
More precisely, define a subset $P_i=\{p\in P: i-k\le\dd_p\le i\}$ for each $i$.
Since a point belongs to $O(k)$ of these subsets, $\sum_i |P_i|=O(kn)$.
For each $p\in P_i$, define 
\begin{align}
    \label{eqn:correction:new}
    \Dhat_i(p) = \min_{q\in P_{i-1}} \max \{ \Dhat_{i-1}(q), \kkappa(p,q) \},
\end{align}
with $\Dhat_0(s)=0$. (All unspecified $\Dhat_i(p)$ values are implicitly set to $\infty$.)

We claim that $r^*=\min_{i\le\lambda}\Dhat_i(t)$.

\subparagraph*{Correctness.}
It is obvious that $\Dhat_i(p)\ge D_i(p)$ for all $p$ and $i$.  Thus, $\min_{i\le\lambda}\Dhat_i(t)\ge 
\min_{i\le\lambda}D_i(t)=r^*$.

For the reverse direction, take a shortest path $\langle u_0, u_1, \ldots, u_i\rangle$ from $s$ to $t$ in $G_{r^*}(P)$,
with $u_0=s$, $u_i=t$, and $i\le \lambda$.
Note that $d_{r^*}(s,u_j)=j$ for every $j=0,\ldots,i$.
So, $j-k\le \dd_{u_j}\le j$, i.e., $u_j\in P_j$ for every $j=0,\ldots,i$.
It follow from (\ref{eqn:correction:new}) that 
$\Dhat_0(u_0)=0$,
$\Dhat_1(u_1)\le r^*$,
$\Dhat_2(u_2)\le r^*$, \ldots, $\Dhat_i(u_i)\le r^*$
(since $\kkappa(u_0,u_1),\ldots,\kkappa(u_{i-1},u_i)\le r^*$).
Thus, $\min_{i\le\lambda}\Dhat_i(t)\le r^*$.

\subparagraph*{Running time.}
For an efficient implementation of the dynamic program, we need a data structure for the following variant of nearest neighbor queries:

\begin{lemma}\label{lem:nnq}
Given a set $Q$ of $n$ points in the plane where each point $q$ has a weight $w_q$,
we can preprocess $Q$ in $O(n\log^2n)$ time so that given any query point $p$,
we can find $q\in Q$ minimizing $\max\{\kkappa(p,q),w_q\}$ in $O(\log^2n)$ time,
where $\kkappa(\cdot,\cdot)$ denotes Euclidean distance.
\end{lemma}
\begin{proof}
We apply standard multi-level geometric data structuring techniques~\cite{AgarwalE99}:
Order $Q$ by increasing weights.  Split $Q$ into the left half $Q_L$ and the right half $Q_R$.
Store a point location structure for the Voronoi diagram of $Q_L$~\cite{BergCKO08}.
Recursively build data structures for $Q_L$ and for $Q_R$.
The preprocessing time satisfies the recurrence $P(n)=2P(n/2)+O(n\log n)$.

Given a query point $p$,
we first find its nearest neighbor $q_L$ in $Q_L$ in $O(\log n)$ time.
\begin{itemize}
\item Case 1: $\kkappa(p,q_L) \le w_{q_L}$.
Then  $\min_{q\in Q_R}\max\{\kkappa(p,q),w_q\}\ge w_{q_L}=\max\{\kkappa(p,q_L),w_{q_L}\}$,
and so it suffices to recurse in $Q_L$.
\item Case 2: $\kkappa(p,q_L) > w_{q_L}$.
Then $\min_{q\in Q_L}\max\{\kkappa(p,q),w_q\} = \kkappa(p,q_L)$,
and so it suffices to recurse in $Q_R$.
\end{itemize}
The query time satisfies the recurrence $Q(n)=Q(n/2)+O(\log n)$, and thus we have $Q(n) = O(\log^2 n)$.
\end{proof}

Having computed $\Dhat_{i-1}(\cdot)$, we can evaluate $\Dhat_i(\cdot)$ according to (\ref{eqn:correction:new}) by preprocessing $P_{i-1}$ in the data structure from Lemma~\ref{lem:nnq} (with weights $w_q:=\Dhat_{i-1}(q)$)
and answering $|P_i|$ queries.
The running time is $\OO(|P_{i-1}|+|P_i|)$.
Summing over all $i$ yields 
total running time $\OO(kn)$.  This completes the proof of Lemma~\ref{lem:correction}.

\begin{theorem}
    The reverse shortest path problem in (unweighted) unit-disk graphs can be solved in $\OO(n^{9/8})$ time w.h.p.\ by a (Las Vegas) randomized algorithm.
\end{theorem}

\begin{proof}
    By combining Lemma~\ref{lem:approx-search} and Lemma~\ref{lem:correction}, we obtain an algorithm that runs in time \begin{align*}
    \OO\left(
        \frac{\Delta^{1/3} n}{L^{2/3}} +  L^{1/2}n + \frac{n^2}{\Delta}
    \right).
\end{align*}
Choosing $L = n^{1/4}$ and $\Delta = n^{7/8}$ yields the desired result.
\end{proof}

\bibliographystyle{plainurl}
\bibliography{reference}

\appendix
\section{Reverse Shortest Path in Visibility Graphs in 1.5-Dimensional Terrains}
\label{sec:terrain}

\begin{figure}
    \centering
    \includegraphics[width=0.8\linewidth]{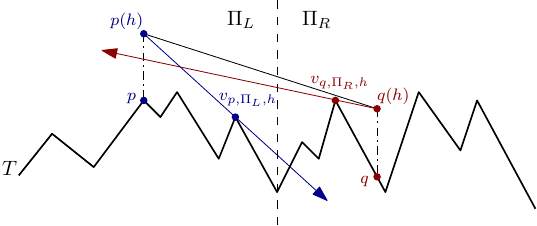}
    \caption{Consider two vertical slabs $\Pi_L$ and $\Pi_R$ divided by a vertical line, two points $p$ and $q$ on $T$ lying in $\Pi_L$ and $\Pi_R$ respectively, and a height $h$. This figure shows the critical rays $\overrightarrow{p(h)v_{p,\Pi_L,h}}$ in blue and $\overrightarrow{q(h)v_{q,\Pi_R,h}}$ in red. In this figure, both critical rays are below the segment $\overline{p(h)q(h)}$, so $p(h)$ and $q(h)$ see each other.}
    \label{fig:terrain}
\end{figure}

In this appendix, we briefly elaborate on the application to reverse shortest path in visibility
graphs in 1.5-dimensional terrains, and justify why $\alpha=2/3$ is attainable as claimed
in the proof of Theorem~\ref{thm:appl}.

Let $h^*$ denote the optimal height.  Let $p(h)$ denote the point $p+(0,h)$.
For two points $p,q\in P$, one could define $\kkappa(p,q)$ as the smallest $h$ such that
$p(h)$ and $q(h)$ are visible.
In Katz \etal's $\OO(n)$-time decision algorithm~\cite{KatzSS24},
the critical values are indeed among such values $\kkappa(p,q)$.
However, the condition that $p(h)$ and $q(h)$ are visible is not a constant-complexity predicate,
being dependent on the terrain $T$, which makes the implementation of the selection oracle more difficult.

To remedy this issue, we recall Katz \etal's algorithm, which uses divide-and-conquer.
The recursion generates a binary tree of vertical slabs such that the tree has $O(\log n)$ height.
For two input points $p$ and $q$ in sibling slabs $\Pi_L$ and $\Pi_R$ respectively, one can find vertices $v_{p,\Pi_L,h}$
and $v_{q,\Pi_R,h}$ with the following property: $p(h)$ and $q(h)$ are visible iff the rays $\overrightarrow{p(h)v_{p,\Pi_L,h}}$
and $\overrightarrow{q(h)v_{q,\Pi_R,h}}$ are both below $\overline{p(h)q(h)}$.
(These rays are called \emph{critical rays} in their paper.) See Figure~\ref{fig:terrain} for an illustration.

Katz \etal~\cite{KatzSS24} described a ``preliminary stage'', which in $\OO(n)$ time determines all these vertices
$v_{p,\Pi_L,h^*}$ and $v_{q,\Pi_R,h^*}$ at the unknown $h^*$ (by using standard parametric search).

Now, we can redefine $\kkappa((p,\Pi_L),(q,\Pi_R))$ as the smallest $h$ such that
the rays $\overrightarrow{p(h)v_{p,\Pi_L,h^*}}$
and $\overrightarrow{q(h)v_{q,\Pi_R,h^*}}$ are both below $\overline{p(h)q(h)}$
(i.e., $q(h)$ is above the line $\overleftrightarrow{p(h)v_{p,\Pi_L,h^*}}$
and $p(h)$ is above the line $\overleftrightarrow{q(h)v_{q,\Pi_R,h^*}}$).
There are $O(n\log n)$ ``objects'' $(p,\Pi_L)$ and $(q,\Pi_R)$ in total (since a point lies in $O(\log n)$ slabs in the tree).
For this function $\kkappa$, one can then implement the selection oracle by known geometric range searching techniques,
i.e., first using standard parametric search to reduce selection to counting, and then using multi-level data structures
involving halfplane range searching in the plane (e.g., \cite{AgarwalE99,ChanZ24}).
This yields an $\OO(|\widehat{P}|^{2/3}|\widehat{Q}|^{2/3} + |\widehat{P}|+|\widehat{Q}|)$ time bound for selection in $\kkappa(\widehat{P},\widehat{Q})$ for any two subsets of ``objects'' $\widehat{P}$ and $\widehat{Q}$.  Thus, $\alpha=2/3$ indeed holds in this application.

\end{document}